\documentclass[a4paper,UKenglish,numberwithinsect]{lipics-v2016}

\usepackage{microtype}

\title{On degeneration of tensors and algebras}
\author[1]{Markus Bl\"aser}
\author[2]{Vladimir Lysikov}
\affil[1]{Department of Computer Science, Saarland University, Saarbr\"ucken, Germany\\
    \texttt{mblaeser@cs.uni-saarland.de}}
\affil[2]{Cluster of Excellence MMCI and Department of Computer Science, Saarland University, Saarbr\"ucken, Germany\\
    \texttt{vlysikov@cs.uni-saarland.de}}
\authorrunning{M. Bl\"aser and V. Lysikov}
\Copyright{Markus Bl\"aser and Vladimir Lysikov}


\EventEditors{Piotr Faliszewski, Anca Muscholl, and Rolf Niedermeier}
\EventNoEds{3}
\EventLongTitle{41st International Symposium on Mathematical Foundations of Computer Science (MFCS 2016)}
\EventShortTitle{MFCS 2016}
\EventAcronym{MFCS}
\EventYear{2016}
\EventDate{August 22--26, 2016}
\EventLocation{Krak\'ow, Poland}
\EventLogo{}
\SeriesVolume{58}
\ArticleNo{XX}
\DOIPrefix{}

\begin{document}

\maketitle

\begin{abstract}
An important building block in all current asymptotically fast algorithms for matrix
multiplication are tensors with low border rank, that is, tensors whose border rank is
equal or very close to their size. To find new asymptotically fast algorithms for 
matrix multiplication, it seems to be important to understand those tensors whose
border rank is as small as possible, so called tensors of minimal border rank.

    We investigate the connection between degenerations of associative algebras and degenerations of their structure tensors in the sense of Strassen.
    It allows us to describe an open subset of $n \times n \times n$ tensors of minimal border rank in terms of smoothability of commutative algebras.
    We describe the smoothable algebra associated to the Coppersmith-Winograd tensor and prove a lower bound for the border rank of the tensor used in the ``easy construction'' of Coppersmith and Winograd.
\end{abstract}

\section{Introduction}

Let $V_1, V_2, V_3$ be vector spaces.
The tensor product $V_1\otimes V_2\otimes V_3$ is spanned by tensors of the form $v_1\otimes v_2\otimes v_3$, which are called \emph{decomposable} tensors,
i.~e., any tensor $T\in V_1\otimes V_2\otimes V_3$ can be represented as a sum 
\begin{equation}
\label{eqn-polyadic-decomposition}
 T = \sum_{s = 1}^r v_{1,s} \otimes v_{2,s} \otimes v_{3,s}.
\end{equation}
This representation is called a \emph{polyadic decomposition} of $T$.
The minimal number of summands in a polyadic decomposition of $T$ is called the \emph{rank} of $T$.
Tensor rank is a direct generalization of the usual notion of matrix rank, which can be defined as a minimal number of summands in a representation of a matrix as a sum of rank one matrices.
Unlike in the matrix case, the set $\mathbf{R}_r$ of all tensors of rank at most $r$ is in general not closed,
so it is useful to consider not only exact polyadic decompositions, but also approximations of tensors by sums of the form~\eqref{eqn-polyadic-decomposition}.
Given a tensor $T$, the minimal number $r$ such that $T$ is contained in the closure of $\mathbf{R}_r$ is called the \emph{border rank} of $T$.

Rank and border rank of tensors have diverse applications (see~\cite{comon-2002-survey,kolda-bader-2009-survey,landsberg-2012-book} for more information).
Our motivation originates from computational complexity theory of bilinear maps.
Any bilinear map $\varphi\colon U \times V \to W$ between finite-dimensional vector spaces is a contraction with some tensor from $U^* \otimes V^* \otimes W$, called the \emph{structural tensor} of $\varphi$.
Polyadic decompositions of the structural tensor can be interpreted as algorithms of a certain kind for computing $\varphi$,
and the rank of the structural tensor of a bilinear map is a measure of its computational complexity.
See~\cite{buergisser-1997-book} for a detailed exposition of bilinear complexity theory.

One interesting problem in this area is the classification of concise tensors of minimal border rank.
A tensor from $V_1 \otimes V_2 \otimes V_3$ is \emph{concise} if it is not contained in any proper subspace $V'_1 \otimes V'_2 \otimes V'_3 \subset V_1 \otimes V_2 \otimes V_3$.
The border rank of a concise tensor is bounded from below by $\max \{ \dim V_i \}$. Tensors for which this bound is exact are called tensors \emph{of minimal border rank}.

Tensors of minimal border rank correspond to bilinear maps that have low complexity and can be used to construct efficient bilinear algorithms.
For example, the famous Coppersmith-Winograd algorithm for matrix multiplication~\cite{coppersmith-winograd-1990-matrix} (as well as its recent improvements~\cite{legall-2014-powers}) uses a tensor of minimal border rank as a basic block.
Such a tensor, like the Coppersmith-Winograd tensor, usually appears ``out of the blue'' and this starting tensor, which at a first glance has only very little to do with matrix multiplication, is then used to design a fast  matrix multiplication algorithm by looking at high powers of the starting tensor.
Therefore, to make further progress in the design of fast matrix multiplication algorithms, a systematic description of the tensors of minimal border rank seems to be very helpful.
As our first result, we describe (an open subset) of the tensors of minimal border rank in terms of their structure.
It turns out, that these tensors are the multiplication tensors of so-called smoothable commutative algebras.
These algebras are studied in algebraic geometry in connection with Hilbert schemes of points, and have received quite some attention in the recent years, however, their structure is not fully understood.

Recently, Landsberg and Micha\l{}ek~\cite{landsberg-michalek-2015-abelian} described tensors of minimal border rank in $\mathbb{C}^n \otimes \mathbb{C}^n \otimes \mathbb{C}^n$ that have a slice of rank $n$ in terms of certain Lie algebras constructed from the slices of the tensor.
In this paper we consider a slightly stronger condition, namely, existence of rank-$n$ slices in two slicing directions,
and prove that any tensor of minimal border rank satisfying this condition is equivalent to a structure tensor of a smoothable commutative algebra.

Furthermore, we describe a method which can be though of as a limiting version of the substitution method for the rank lower bounds (essentially the same method was independently described by Landsberg and Micha\l{}ek~\cite{landsberg-michalek-2016-geometry}) and use it to prove a lower bound on the border rank of tensor powers of the restricted Coppersmith-Winograd tensor used in the ``easy construction'' of~\cite{coppersmith-winograd-1990-matrix}.
This easy tensor is not a tensor of minimal border rank.
However, as pointed out  in~\cite[Rem.~15.44]{buergisser-1997-book}, if this tensor had asymptotically minimal  border rank, then the exponent of matrix multiplication would be $2$. 
Asymptotically minimal border rank means that the border rank of the tensor powers converges to the size of the tensor.
While our bound is nontrivial, it does not rule out that the easy tensor has asymptotically minimal border rank. 

\section{Preliminaries}

\subsection{Notation and basic definitions}

All vector spaces are presumed to be finite-dimensional vector spaces over some fixed algebraically closed field $k$.
Letters $U$, $V$, $W$, possibly indexed, denote vector spaces, $A$ denotes algebras.
$\varepsilon$ usually denotes some indeterminate.
Linear maps and tensors over $k(\varepsilon)$ are rendered in calligraphic font.

We do not distinguish between bilinear maps $U \times V \to W$ and corresponding tensors in $U^* \otimes V^* \otimes W$.
When there is no confusion, multiplication in some algebra as a bilinear map is denoted by the same symbol as the algebra itself.
In particular, $k^r$ denotes the coordinate-wise multiplication of $r$-dimensional vectors.

A bilinear map $\varphi \in V^* \otimes V^* \otimes V$ is called \emph{unital} if there is an \emph{identity element} $e \in V$ such that $\varphi(e, x) = \varphi(x, e) = x$ for all $x\in V$.

Any tensor in $V_1 \otimes V_2 \otimes V_3$ with $\dim V_i = n_i$ is said to have \emph{format} $n_1 \times n_2 \times n_3$. We are mostly interested in tensors of format $n \times n \times n$.

Rank and border rank of a tensor $T$ are denoted by $R(T)$ and $\underline{R}(T)$ respectively.

The Zariski closure of a set $\mathbf{S}$ is denoted by $\overline{\mathbf{S}}$. We use the Zariski topology to define the border rank of tensors over $k$. However, over $\mathbb{C}$, the Zariski closure of the set of all tensors in $V_1 \otimes V_2 \otimes V_3$ of rank at most $r$ coincides with its Euclidean closure, thus capturing the idea of approximation.

Let $T\in V_1\otimes V_2\otimes V_3$ and $T'\in V'_1\otimes V'_2\otimes V'_3$ be two tensors.
$T'$ is a \emph{restriction} of $T$ (denoted $T'\leq T$) if there exists a triple of linear maps $F_i\colon V_i\to V'_i$ such that $T' = (F_1\otimes F_2\otimes F_3) T$. We call the operator $(F_1\otimes F_2\otimes F_3)$ \emph{a restriction operator} for $T' \leq T$.

Two tensors $T_1$ and $T_2$ are called \emph{equivalent} ($T_1 \sim T_2$) if $T_1 \leq T_2$ and $T_2 \leq T_1$. If $T_1$ and $T_2$ have the same format, they are equivalent iff there is a bijective restriction operator for $T_1 \leq T_2$.

The tensor rank can be defined via restrictions of $k^r$ (or, equivalently, diagonal tensors $\sum_{i = 1}^r e_i \otimes e_i \otimes e_i$): polyadic decompositions of a tensor $T$ are in one-to-one correspondence with restriction operators for $T \leq k^r$, so $R(T) \leq r$ iff $T \leq k^r$.

For more information, we refer to~\cite{buergisser-1997-book}.

\subsection{Degeneration of tensors}

Degeneration of tensors was introduced by Strassen~\cite{strassen-1987-relative}.
It is an approximate analogue of restriction:
a tensor $T'$ is a \emph{degeneration} of $T$ (denoted $T' \trianglelefteq T$) if \[T' \in \overline{\{t\in V'_1\otimes V'_2\otimes V_3 \mid t \leq T\}}.\] 

Strassen gives alternative descriptions of degeneration.
One of these descriptions is in terms of representation theory.
Consider the group $G  = \operatorname{GL}(V_1) \times \operatorname{GL(V_2)} \times \operatorname{GL}(V_3)$.
It acts on $V_1 \otimes V_2 \otimes V_3$ in a standard way: \[(F_1, F_2, F_3)  \cdot T = (F_1 \otimes F_2 \otimes F_3) T.\]
The orbits of this action are equivalence classes of tensors in $V_1 \otimes V_2 \otimes V_3$.

\begin{lemma}[Strassen~\cite{strassen-1987-relative}]
    Let $T\in V_1\otimes V_2\otimes V_3$ and $T'\in V'_1\otimes V'_2\otimes V'_3$ be two tensors. $T' \trianglelefteq T$ if and only if there exists a tensor $S \in V_1\otimes V_2\otimes V_3$ such that $T'\sim S$ and $S \in \overline{G\cdot T}$.
\end{lemma}

Another description of degeneration uses base extension from $k$ to $k(\varepsilon)$.
For any vector space $V$ over $k$ its base extension $V(\varepsilon) = V \otimes k(\varepsilon)$ is a vector space over $k(\varepsilon)$.
We have an injection $V \hookrightarrow V(\varepsilon)$ defined by $v \mapsto v \otimes 1$.
Analogously, $V_1 \otimes V_2 \otimes V_3$ injects into $V_1(\varepsilon) \otimes_{k(\varepsilon)} V_2(\varepsilon) \otimes_{k(\varepsilon)} V_3(\varepsilon)$, so any tensor $T\in V_1 \otimes V_2 \otimes V_3$ can be viewed as a tensor over $k(\varepsilon)$, which we also denote by $T$.

Going in the other direction, we have a partial map from $k(\varepsilon)$ to $k$ that takes a rational function regular at $\varepsilon = 0$ to its value at $0$.
It can be extended to partial maps from $V(\varepsilon)$ to $V$ for each vector space $V$.
If $\mathcal{T}|_{\varepsilon = 0} = T$, we sometimes write $\mathcal{T} = T + O(\varepsilon)$, thinking of $\varepsilon$ as an infinitesimal.

\begin{lemma}[Strassen~\cite{strassen-1987-relative}]
\label{lemma-algebraic-definition}
    Let $T\in V_1\otimes V_2\otimes V_3$ and $T'\in V'_1\otimes V'_2\otimes V'_3$ be two tensors.
    $T' \trianglelefteq T$ if and only if there exists $\mathcal{T} \in V_1(\varepsilon) \otimes_{k(\varepsilon)} V_2(\varepsilon) \otimes_{k(\varepsilon)} V_3(\varepsilon)$ such that $\mathcal{T}|_{\varepsilon = 0} = T'$ and $\mathcal{T} \leq T$ as tensors over $k(\varepsilon)$.
\end{lemma}

This lemma allows us to talk about specific ways in which $T$ degenerates into $T'$, which are represented by restriction operators for restrictions of the form $T' + O(\varepsilon) \leq T$ considered in the lemma.
We call these operators \emph{degeneration operators} for $T' \trianglelefteq T$.

Degenerations of $k^r$ are related to border rank in the same way its restrictions are related to rank:
since $R(\varphi) \leq r$ iff $\varphi \leq k^r$, by taking closures we have $\underline{R}(\varphi) \leq r$ iff $\varphi \trianglelefteq k^r$.
In particular, Lemma~\ref{lemma-algebraic-definition} implies existence of \emph{approximate polyadic decompositions}
\[ T + O(\varepsilon) = \sum_{s = 1}^r \mathcal{\scriptstyle V}_{1,s} \otimes_{k(\varepsilon)} \mathcal{\scriptstyle V}_{2,s} \otimes_{k(\varepsilon)} \mathcal{\scriptstyle V}_{3,s}\]

\subsection{Degeneration of algebras}

Strassen's theory of tensor degenerations was inspired by the similar concept in the deformation theory of algebras.

Degeneration of algebras is usually restricted to associative or Lie algebras, but we define it for arbitrary bilinear maps in $V^* \otimes V^* \otimes V$, which can be thought of as nonassociative algebra structures on $V$.
The group $\operatorname{GL}(V)$ acts on $V^* \otimes V^* \otimes V$ by change of basis: \[(g\cdot \varphi) (x, y) = g \varphi(g^{-1} x, g^{-1} y).\]
The orbits of this action are isomorphism classes of nonassociative algebras.

Let $\varphi, \varphi' \in V^* \otimes V^* \otimes V$. We call $\varphi'$ an \emph{algebraic degeneration} of $\varphi$ (denoted $\varphi' \trianglelefteq_{\mathrm{a}} \varphi$) if $\varphi'$ lies in the orbit closure $\overline{\operatorname{GL}(V) \cdot \varphi}$.
The name ``algebraic degeneration'' is used here to distinguish between two different notions of degeneration on $V^* \otimes V^* \otimes V$ and does not appear in the literature on degeneration of algebras.
It is easy to see that $\varphi' \trianglelefteq_{\mathrm{a}} \varphi$ implies $\varphi' \trianglelefteq \varphi$.

We can extend the definition of algebraic degeneration to bilinear maps on different spaces of the same dimension by saying that if $\varphi'$ is an algebraic degeneration of $\varphi$, then any $\psi'$ isomorphic to $\varphi'$ as a nonassociative algebra is also a algebraic degeneration of $\varphi$.

Since associativity and commutativity properties define closed subsets of $V^* \otimes V^* \otimes V$,
degenerations of associative (resp.~commutative) algebras are themselves associative (commutative).

\begin{definition}
    An unital algebra $A$ of dimension $n$ such that $A \trianglelefteq_{\mathrm{a}} k^n$  is called \emph{smoothable}.
\end{definition}

As follows from previous discussion, smoothable algebras are always associative and commutative.

In the geometric study of finite-dimensional commutative algebras they are sometimes studied as elements of a variety in $V^* \otimes V^* \otimes V$ or a similar scheme,
and sometimes --- as elements of a Hilbert scheme of points $\mathbf{Hilb}_n(\mathbb{A}^d_k)$, which parameterizes $0$-dimensional schemes on $d$-dimensional affine plane,
or, equivalently, ideals $I$ in $R = k[x_1,\dots,x_d]$ such that $R/I$ is an $n$-dimensional algebra.
The exact relationship between these two approaches is explored in~\cite{poonen-2008-moduli}.
We will only need the fact that topologies on $V^* \otimes V^* \otimes V$ and on $\mathbf{Hilb}_n(\mathbb{A}^d_k)$ give the same notion of smoothability, so we can use results from~\cite{cartwright-erman-velasco-viray-2009-hilbert,erman-velasco-2010-syzygetic} formulated in the language of Hilbert schemes.

There are analogues of Lemma~\ref{lemma-algebraic-definition} for algebraic degeneration (for example,~\cite[\S~3.9]{kraft-1982-geometric} gives a geometric formulation of a similar statement).
We only need the easier part of the equivalence which says that if $\varphi'$ is approximated by bilinear maps isomorphic to $\varphi$, then it is an algebraic degeneration of $\varphi$.

\begin{lemma}
\label{lemma-algebraic}
    Let $\varphi, \varphi' \in V^* \otimes V^* \otimes V$ be two bilinear maps on $V$.
    If there exists an invertible $k(\varepsilon)$-linear map $\mathcal{F}\colon V(\varepsilon) \to V(\varepsilon)$ such that
    \[\mathcal{F}^{-1} \varphi(\mathcal{F} x, \mathcal{F} y) |_{\varepsilon = 0} = \varphi'(x, y)\quad \text{for all $x,y \in V$},\]
    then $\varphi' \trianglelefteq_{\mathrm{a}} \varphi$.
\end{lemma}
\begin{proof}
    As $\varepsilon$ varies, the bilinear map $\varphi^{\varepsilon}(x, y) = \mathcal{F}^{-1} \varphi(\mathcal{F} x, \mathcal{F} y)$ traces an algebraic curve in $V^* \otimes V^* \otimes V$.
    Since $\mathcal{F}$ is invertible, its values for Zariski almost all $\varepsilon$ are also invertible, so an open subset of the curve $\{\varphi^{\varepsilon}\}$ lies in the orbit $\operatorname{GL}(V)\cdot \varphi$.
    Therefore, the value at $\varepsilon = 0$ lies in the closure of this orbit.
\end{proof}

We can rephrase this lemma as follows: tensor degeneration $\varphi' \trianglelefteq \varphi$ with a degeneration operator of the form $\mathcal{F}^* \otimes \mathcal{F}^* \otimes \mathcal{F}^{-1}$ implies algebraic degeneration $\varphi' \trianglelefteq_{\mathrm{a}} \varphi$.

\section{Degenerations of associative algebras}

In this section and later \emph{algebra} means associative unital algebra over $k$.

\subsection{Transformations of degeneration operators}

Suppose $T\in V_1 \otimes V_2 \otimes V_3$ and $T' \in V'_1 \otimes V'_2 \otimes V'_3$ are two tensors such that $T' \trianglelefteq T$. Denote by $D(T' \trianglelefteq T)$ the set of all degeneration operators for $T' \trianglelefteq T$.

Let us describe some groups that act on $D(T' \trianglelefteq T)$.
These groups are subgroups of $\operatorname{GL}(V_1(\varepsilon))\times
\operatorname{GL}(V_2(\varepsilon))\times \operatorname{GL}(V_3(\varepsilon))$ and $\operatorname{GL}(V'_1(\varepsilon)) \times \operatorname{GL}(V'_2(\varepsilon)) \times \operatorname{GL}(V'_3(\varepsilon))$ which act on the domain and image of operators in $D(T' \trianglelefteq T)$ in the usual way (a triple $(\mathcal{F}_1, \mathcal{F}_2, \mathcal{F}_3)$ acts via $\mathcal{F}_1 \otimes \mathcal{F}_2 \otimes \mathcal{F}_3$).

Let $T \in V_1 \otimes V_2 \otimes V_3$ be a tensor.
Its \emph{isotropy group} $\Gamma(T)$ is defined as the subgroup of $\operatorname{GL}(V_1) \times \operatorname{GL}(V_2) \times \operatorname{GL}(V_3)$ which leaves $T$ fixed.
Isotropy groups of bilinear maps and their action on the set of all bilinear algorithms were studied by de Groote~\cite{degroote-1978-on-varieties}.
Similarly, we define the \emph{$\varepsilon$-isotropy group} $\Gamma^{\varepsilon}(T)$ of $T$ as the subgroup of $\operatorname{GL}(V_1(\varepsilon)) \times \operatorname{GL}(V_2(\varepsilon)) \times \operatorname{GL}(V_3(\varepsilon))$ that fixes $T$ considered as a tensor over $k(\varepsilon)$.

Suppose $\mathcal{F} \in \operatorname{GL}(V(\varepsilon))$ is a $k(\varepsilon)$-linear map such that $\mathcal{F} = \operatorname{id} + O(\varepsilon)$.
Then for each $\mathcal{\scriptstyle V} \in V(\varepsilon)$ we have $\mathcal{F}\mathcal{\scriptstyle V} |_{\varepsilon = 0} = \mathcal{\scriptstyle V}|_{\varepsilon = 0}$, when 
$\mathcal{\scriptstyle V}|_{\varepsilon = 0}$ is defined.
Let $E(V_1, V_2, V_3)$ be the subgroup of $\operatorname{GL}(V_1(\varepsilon)) \times \operatorname{GL}(V_2(\varepsilon)) \times \operatorname{GL}(V_3(\varepsilon))$ consisting of all triples of such operators.

\begin{lemma}
    The groups $\Gamma(T')$ and $E(V'_1, V'_2, V'_3)$ act on $D(T' \trianglelefteq T)$ on the left and $\Gamma^{\varepsilon}(T)$ acts on the right via composition.
\end{lemma}
\begin{proof}
    Let $\mathcal{F} = \mathcal{F}_1 \otimes \mathcal{F}_2 \otimes \mathcal{F}_3$ be a degeneration operator for $T' \trianglelefteq T$, i.~e., $T' + O(\varepsilon) = \mathcal{F} T$. The described actions preserve this relation, since if $G \in \Gamma(T')$, then $GT' = T'$ and $G(O(\varepsilon)) = O(\varepsilon)$; if $\mathcal{G} \in E(V'_1, V'_2, V'_3)$, then $\mathcal{G}(T' + O(\varepsilon)) = T' + O(\varepsilon)$; and if $\mathcal{G} \in \Gamma^{\varepsilon}(T)$, then $\mathcal{G}T = T$.
\end{proof}

We use these transformations in case when $T$ is the structure tensor of some algebra.
Suppose $A$ is an algebra and $a, b, c$ are three invertible elements of $A$. Let $L_x$ and $R_x$ denote left and right multiplication by $x$ respectively.
Then $({(L_a R_b)}^*, {(L_{b}^{-1} R_c)}^*, L_a^{-1} R_c^{-1})$ is an element of the isotropy group $\Gamma(A)$ arising from the identity $xy = a^{-1}(axb)(b^{-1}yc)c^{-1}$.
The use of this identity is sometimes called \emph{sandwiching} in the literature.
Since the tensor over $k(\varepsilon)$ corresponding to $A$ is $A(\varepsilon) = A \otimes k(\varepsilon)$, an analogous expression with $a, b, c \in A(\varepsilon)$ can be used to construct elements of $\Gamma^{\varepsilon}(A)$.

\subsection{Main theorem}

\begin{theorem}
\label{thm-main}
    Let $A$ be an algebra and $\varphi \in A^* \otimes A^* \otimes A$ be a unital bilinear map.
    Then $\varphi \trianglelefteq A$ iff $\varphi \trianglelefteq_{\mathrm{a}} A$.
\end{theorem}
\begin{proof}
    The implication $\varphi \trianglelefteq_{\mathrm{a}} A \Rightarrow \varphi \trianglelefteq A$ is obvious.
    Let us prove the opposite implication.
    
    Let $\varphi \trianglelefteq A$ and $\mathcal{F}^* \otimes \mathcal{G}^* \otimes \mathcal{H}$ be a degeneration operator, i.~e.,
    \[ \varphi(x,y) = \mathcal{H}(\mathcal{F} x \cdot \mathcal{G} y) |_{\varepsilon = 0}\quad \text{for all $x,y\in A$},\]
    where the multiplication is in $A \otimes k(\varepsilon)$.
    
    Let $e$ be the identity element of $\varphi$.
    After the substitution $x = e$ we have
    \[ y = \varphi(e, y) = \mathcal{H}(\mathcal{F} e \cdot \mathcal{G} y) |_{\varepsilon = 0} = \mathcal{H} L_{\mathcal{F} e} \mathcal{G} y |_{\varepsilon = 0}\quad \text{for all $y\in A$},\]
    so $\mathcal{Q} := \mathcal{H} L_{\mathcal{F} e} \mathcal{G} = \operatorname{id} + O(\varepsilon)$.
    Applying $(\operatorname{id}, {(\mathcal{Q}^{-1})}^*, \operatorname{id}) \in E(A^*, A^*, A)$ to the degeneration operator $\mathcal{F}^* \otimes \mathcal{G}^* \otimes \mathcal{H}$,
    we obtain a new degeneration operator $\mathcal{F}^* \otimes \hat{\mathcal{G}}^* \otimes \mathcal{H}$
    where $\hat{\mathcal{G}} = \mathcal{G} \mathcal{Q}^{-1} = L_{\mathcal{F} e}^{-1} \mathcal{H}^{-1}$.
    
    Analogously, setting $y = e$ we get that $\mathcal{P} := \mathcal{H} R_{\hat{\mathcal{G}} e} \mathcal{F} = \operatorname{id} + O(\varepsilon)$
    and using transformation $({(\mathcal{P}^{-1})}^*, \operatorname{id}, \operatorname{id}) \in E(A^*, A^*, A)$
    we get another degeneration operator $\hat{\mathcal{F}}^* \otimes \hat{\mathcal{G}}^* \otimes \mathcal{H}$
    where $\hat{\mathcal{F}} = \mathcal{F} \mathcal{P}^{-1} = R_{\hat{\mathcal{G}} e}^{-1} \mathcal{H}^{-1}$.
    
    Finally, we use a sandwiching transformation $({(L_{\mathcal{F}e}^{-1})}^*, {(R_{\hat{\mathcal{G}} e}^{-1})}^* , L_{\mathcal{F}e} R_{\hat{\mathcal{G}} e})$ from $\Gamma^{\varepsilon}(A)$
    and obtain a degeneration operator $\mathcal{S}^* \otimes \mathcal{S}^* \otimes \mathcal{S}^{-1}$
    where \[\mathcal{S} = {(\mathcal{H} L_{\mathcal{F}e} R_{\hat{\mathcal{G}}
    e})}^{-1}.\]
    By Lemma~\ref{lemma-algebraic} we have an algebraic degeneration $\varphi \trianglelefteq_{\mathrm{a}} A$.
\end{proof}

This theorem can be seen as an extension of the fact that associative algebras have equivalent structure tensors iff they are isomorphic (\cite[Prop.~14.13]{buergisser-1997-book}). The general idea of the proof --- using symmetries of the tensors to transform maps that express the relationship between them --- goes back to de Groote~\cite{degroote-1978-on-varieties}, but in our case some care needed to track the behaviour of degeneration operators as $\varepsilon$ varies.

\subsection{Tensors of minimal border rank}

A special case of Theorem~\ref{thm-main} when the algebra $A$ is $k^r$ can be used to study tensors of minimal border rank. First, we describe algebras of minimal border rank:

\begin{corollary}
    A unital bilinear map on a vector space of dimension $n$ is of minimal border rank iff it is a multiplication in a smoothable algebra.
\end{corollary}
\begin{proof}
    By Theorem~\ref{thm-main} in the present case it is equivalent to $\varphi \trianglelefteq_{\mathrm{a}} k^n$, which is the definition of a smoothable algebra.
\end{proof}

For example, if $\operatorname{char} k \neq 2, 3$,
the following algebras are smoothable~\cite{cartwright-erman-velasco-viray-2009-hilbert}, and, therefore, have minimal border rank:
\begin{enumerate}
    \item any algebra generated by $2$ elements;
    \item any algebra of the form $k[x_1,\dots,x_d]/I$ where the ideal $I$ is monomial;
    \item any algebra with $\dim (R^2/R^3) = 1$ where $R = \operatorname{rad} A$;
    \item any algebra with $\dim (R^2/R^3) = 2$, $\dim R^3 \leq 2$ and $R^4 = 0$ where $R = \operatorname{rad} A$;
    \item any algebra of dimension $7$ or less;
\end{enumerate}

A description of smoothable algebras of dimension $8$ is contained in~\cite{cartwright-erman-velasco-viray-2009-hilbert,erman-velasco-2010-syzygetic}.

Using the description of algebras of minimal border rank, we can identify a certain open subset of tensors of minimal border rank.

\begin{definition}
\label{def-binding}
    A tensor $T \in V_1 \otimes V_2 \otimes V_3$ of format $n \times n \times n$ is called \emph{binding} if there are elements $\alpha_1 \in V_1^*$ and $\alpha_2 \in V_2^*$ such that the contractions $T \alpha_1 \in V_2 \otimes V_3$ and $T \alpha_2 \in V_1\otimes V_2$ have rank $n$.
\end{definition}

Note that a generic tensor of format $n \times n \times n$ is binding.
In the terminology of~\cite{landsberg-michalek-2015-abelian} binding tensors are
called $1_{V_1}$- and $1_{V_2}$-generic.
We call these tensors binding because they allow us to relate spaces $V_1$ and $V_2$ to $V_3$ similarly to how a nondegenerate bilinear form allows to view spaces of its arguments as dual to each other. This is used in the proof of the following lemma.

\begin{lemma}
\label{lemma-binding}
    A binding tensor is equivalent to an unital bilinear map.
\end{lemma}
\begin{proof}
    Let $\dim V_1 = \dim V_2 = \dim V_3 = n$ and $T\in V_1\otimes V_2\otimes V_3$ be a binding tensor. Let $\alpha_1 \in V_1^*$ and $\alpha_2 \in V_2^*$ be as in Definition~\ref{def-binding}.
    
    We can view $T \alpha_1$ and $T \alpha_2$ as linear isomorphisms $P_1 \colon V_1^* \to V_3$ and $P_2 \colon V_2^* \to V_3$.
    Applying ${(P_2^{-1})}^* \otimes {(P_1^{-1})}^* \otimes \operatorname{id}$ to $T$
    we get an equivalent bilinear map
    \[\varphi(x_1, x_2) =  T(P_2^{-1} x_1)(P_1^{-1} x_2).\]
    This bilinear map is unital, since $\varphi(P_2 \alpha_1, x) = x$ and $\varphi(x, P_1\alpha_2) = x$ for all $x\in V_1$,
    so \[P_2\alpha_1 = \varphi(P_2\alpha_1, P_1\alpha_2) = P_1\alpha_2\] is the identity element.
\end{proof}

\begin{corollary}
\label{cor-binding}
    A binding tensor has minimal border rank iff it is equivalent to a smoothable algebra.
\end{corollary}

These results suggest that structure tensors of smoothable algebras are possible candidates for basic blocks to construct fast matrix multiplication algorithms. We tried to use some of them in the same framework that is used by Coppersmith and Winograd (it is known as ``laser method'', see~\cite{buergisser-1997-book,legall-2014-powers} for more information). So far, these attempts did not lead to improved matrix multiplication algorithms.

\subsection{Example: Coppersmith-Winograd tensor}
\label{section-cw}

Let $e^{[0]}, e^{[1]}_1, \dots, e^{[1]}_q, e^{[2]}$ be a basis of a $(q + 2)$-dimensional vector space,
and $\alpha^{[0]}, \alpha^{[1]}_i, \alpha^{[2]}$ be the dual basis. The famous Coppersmith-Winograd algorithm~\cite{coppersmith-winograd-1990-matrix} uses the tensor
\begin{equation*}
\begin{aligned}
T_{CW} = \sum_{i = 1}^q
( & e^{[0]} \otimes e^{[1]}_i \otimes e^{[1]}_i +
  e^{[1]}_i \otimes e^{[0]} \otimes e^{[1]}_i + 
  e^{[1]}_i \otimes e^{[1]}_i \otimes e^{[0]}) + \\
+ & e^{[0]} \otimes e^{[0]} \otimes e^{[2]} +
e^{[0]} \otimes e^{[2]} \otimes e^{[0]} +
e^{[2]} \otimes e^{[2]} \otimes e^{[0]},
\end{aligned}
\end{equation*}
which we will call \emph{Coppersmith-Winograd tensor}.

We can use the results of the previous section to exhibit a smoothable algebra with the structure tensor equivalent to the Coppersmith-Winograd tensor.

The Coppersmith-Winograd tensor is a tensor of minimal border rank,
as witnessed by the approximate decomposition
\begin{equation}
\begin{aligned}
\label{eqn-cw}
T_{CW} & + O(\varepsilon) =
\varepsilon^{-2} \sum_{i = 1}^q
(e^{[0]} + \varepsilon e^{[1]}_i) \otimes
(e^{[0]} + \varepsilon e^{[1]}_i) \otimes
(e^{[0]} + \varepsilon e^{[1]}_i) - \\
& - \varepsilon^{-3}
(e^{[0]} + \varepsilon^2 \sum_{i = 1}^q e^{[1]}_i) \otimes
(e^{[0]} + \varepsilon^2 \sum_{i = 1}^q e^{[1]}_i) \otimes
(e^{[0]} + \varepsilon^2 \sum_{i = 1}^q e^{[1]}_i) + \\
& + (\varepsilon^{-3} - q \varepsilon^{-2})
(e^{[0]} + \varepsilon^3 e^{[2]}) \otimes
(e^{[0]} + \varepsilon^3 e^{[2]}) \otimes
(e^{[0]} + \varepsilon^3 e^{[2]}).
\end{aligned}
\end{equation}

The Coppersmith-Winograd tensor $T_{CW}$ is binding (the layers corresponding to $\alpha^{[0]}$ have full rank). Applying Lemma~\ref{lemma-binding}, we obtain
a bilinear map
\begin{equation*}
\begin{aligned}
    \sum_{i = 1}^q (&\alpha^{[2]} \otimes \alpha^{[1]}_i \otimes e^{[1]}_i +
  \alpha^{[1]}_i \otimes \alpha^{[2]} \otimes e^{[1]}_i + \alpha^{[1]}_i
  \otimes \alpha^{[1]}_i \otimes e^{[0]}) + \\ + &\alpha^{[2]} \otimes
  \alpha^{[2]} \otimes e^{[2]} + \alpha^{[2]} \otimes \alpha^{[0]} \otimes
  e^{[0]} + \alpha^{[0]} \otimes \alpha^{[0]} \otimes e^{[0]}
\end{aligned}
\end{equation*}
which is unital with the identity $e^{[2]}$.
By Corollary~\ref{cor-binding} this map is a multiplication in some smoothable algebra.
Denote $e^{[2]}$ by $1$ and $e^{[1]}_i$ by $x_i$.
In this notation, $x_i x_j = 0$ for $i\neq j$ and $e^{[0]}$ corresponds to $x_1^2 = x_2^2 = \dots = x_q^2$. To summarize,

\begin{example}
    The Coppersmith-Winograd tensors is equivalent to the smoothable algebra $A_{CW} \cong k[x_1,\dots,x_q]/\left<x_i x_j, x_i^2 - x_j^2, x_i^3 \mid i\neq j\right>$.
\end{example}

Performing transformations described in the proof of Theorem~\ref{thm-main} for
the decomposition~\eqref{eqn-cw}, we can construct an algebraic degeneration of
$k^{d + 2}$ to $A_{CW}$ given by a degeneration operator $\mathcal{S}^{*}
\otimes \mathcal{S}^{*} \otimes \mathcal{S}^{-1}$ where $\mathcal{S} \colon
A_{CW}(\varepsilon) \to {k(\varepsilon)}^{q+2}$ has the following matrix
relative to the basis $\{1, x_1, \dots, x_q, x_1^2\}$ in $A_{CW}$ and the
standard basis in $k^{q + 2}$:
\begin{equation*}
    \scriptstyle
    \left[ \begin{array}{c|ccccc|c}
        1 & \varepsilon - (q - 1) \varepsilon^2 & \varepsilon^2 & \varepsilon^2 & \cdots & \varepsilon^2 & -\varepsilon^3 \\
        1 & \varepsilon^2 & \varepsilon - (q - 1) \varepsilon^2 & \varepsilon^2 & \cdots & \varepsilon^2 & -\varepsilon^3 \\
        1 & \varepsilon^2 & \varepsilon^2 & \varepsilon - (q - 1) \varepsilon^2 & \cdots & \varepsilon^2 & -\varepsilon^3 \\
        \vdots & \vdots & \vdots & \vdots & \ddots & \vdots & \vdots \\
        1 & \varepsilon^2 & \varepsilon^2 & \varepsilon^2 & \cdots & \varepsilon - (q - 1) \varepsilon^2 & -\varepsilon^3 \\
        1 & \varepsilon^2 & \varepsilon^2 & \varepsilon^2 & \cdots & \varepsilon^2 & - \varepsilon^3 \\
        1 & 0 & 0 & 0 & \cdots & 0 & 0
    \end{array} \right].
\end{equation*}
We may simplify this matrix by applying a certain linear map of the form
$\operatorname{id} + O(\varepsilon)$, obtaining a new degeneration
corresponding to a matrix
\begin{equation}
\label{eqn-matrix}
    \left[ \begin{array}{c|ccccc|c}
    1 & \varepsilon & 0 & 0 & \cdots & 0 & -\varepsilon^3 \\
    1 & 0 & \varepsilon & 0 & \cdots & 0 & -\varepsilon^3 \\
    1 & 0 & 0 & \varepsilon & \cdots & 0 & -\varepsilon^3 \\
    \vdots & \vdots & \vdots & \vdots & \ddots & \vdots & \vdots \\
    1 & 0 & 0 & 0 & \cdots & \varepsilon & -\varepsilon^3 \\
    1 & \varepsilon^2 & \varepsilon^2 & \varepsilon^2 & \cdots & \varepsilon^2 & - \varepsilon^3 \\
    1 & 0 & 0 & 0 & \cdots & 0 & 0
    \end{array} \right].
\end{equation}

In the language of schemes this degeneration can be interpreted as follows:
the $0$-dimensional scheme $\mathbf{S}_{CW}$ with coordinate ring $A_{CW}$ is the flat limit of the family (parameterized by $\varepsilon$) of schemes containing $q + 2$ points in $(q+2)$-dimensional affine space with coordinates given by
the rows of the matrix~\eqref{eqn-matrix}.

Since $A_{CW}$ is generated by $q$ elements, $\mathbf{S}_{CW}$ is contained in a $q$-dimensional affine subspace, so we can consider instead of schemes in $(q + 2)$-dimensional space their projections to this subspace, which corresponds to the middle part of~\eqref{eqn-matrix}.

For those unfamiliar with the terminology of schemes, here is an algorithmic interpretation:
to approximately multiply two elements of $A_{CW}$,
evaluate the corresponding polynomials of the form $a^{[0]} + \sum a^{[1]}_i x_i + a^{[2]} x_1^2$ at the $q + 2$ points given by the middle part of the matrix, multiply the corresponding values, and interpolate the products to get a resulting polynomial.

\section{Substitution method for border rank}

In this section we describe a method for obtaining lower bounds which can be seen as a border rank version of the substitution method for tensor rank.
Let $T\in U^* \otimes V \otimes W$.
We can view it as a linear map $U\to V\otimes W$ and consider the restriction $T|_{U'} \in (U')^* \otimes V \otimes W$ for any subspace $U' \subset U$. If the border ranks of $T|_{U'}$ are known, we can derive the bound on the border rank of $T$.

\begin{theorem}
    Let $T\in U^* \otimes V \otimes W$ and $\dim U = n$. For any $d$ we have
    \[\underline{R}(T) \geq n - d + \min \{ \underline{R}(T|_{U'}) \mid U' \subset U,\,\dim U' = d\}.\]
\end{theorem}
\begin{proof}
    Suppose $\underline{R}(T) = r$.
    We can assume that $T$ is concise, considering it as an element of a smaller subspace $(U')^* \otimes V' \otimes W' \subset U^* \otimes V \otimes W$ otherwise.
    We need to show that there exists a subspace $U' \subset U$, $\dim U' = d$, such that $\underline{R}(T|_{U'}) \leq r - n + d$.

    Note that this is true for tensors $T$ of rank $r$. 
    Indeed, let $T = \sum_{s = 1}^r f_s \otimes v_s \otimes w_s$ be a polyadic decomposition.
    Without loss of generality, $f_1, \dots, f_n$ form a basis of $U^*$,
    and for the $d$-dimensional subspace $U' \subset U$ defined by the equations $f_i = 0$ for~$1 \leq i \leq n - d$
    we have $\underline{R}(T|_{U'}) \leq R(T|_{U'}) \leq r - n + d$, since the first $n - d$ terms of the decomposition vanish on $U'$.
    
    Moreover, if we have an approximate decomposition \[T + O(\varepsilon) = \sum_{s = 1}^r f_s(\varepsilon) \otimes v_s(\varepsilon) \otimes w_s(\varepsilon) = \mathcal{T},\]
    we can assume that $f_1(\varepsilon), \dots, f_n(\varepsilon)$ are linearly independent for almost all values of $\varepsilon$ (because concise tensors form an open set),
    and obtain a family of subspaces $U'_{\varepsilon}$ such that $\mathcal{T}(\varepsilon)|_{U'_{\varepsilon}}$ has rank at most $r - n + d$.
    The family $U'_{\varepsilon}$ defines an algebraic curve in the Grassmannian $\mathbf{Gr}(d, U)$.
    Grassmannians are projective varieties, so $U'_{\varepsilon}$ can be extended to $\varepsilon = 0$ (see, for example,~\cite[Rem.~7.12, Thm.~7.22]{milne-2015-AG}).
    
    Given an isomorphism $F\colon k^d \to U' \subset U$ and a tensor $T$, we can define $\hat{T} \in (k^d)^* \otimes V \otimes W$ as $\hat{T}(p) = T(Fp)$ so that $\hat{T} \sim T|_{U'}$ and the map $Z \colon (T, F) \mapsto \hat{T}$ is algebraic. In the neighborhood of $U'_0$ we can choose isomorphisms $\mathcal{F} \colon k^d \to U'_{\varepsilon}$ which vary continuously with~$\varepsilon$.
    Using these isomorphisms, we include $T|_{U'_0}$ in an algebraic family $Z(\mathcal{T}, \mathcal{F})$ of tensors of rank at most $r - n + d$, therefore, its border rank does not exceed this value.
\end{proof}

Essentially the same method was independently described by Landsberg and Micha\l{}ek~\cite{landsberg-michalek-2016-geometry}.
They prove this lower bound when $U'$ is a hyperplane in~$U$ (from which the general version follows easily) and use it to obtain a lower bound on the rank of matrix multiplication. We consider the other extremal case where $U' = \left<u\right>$ is $1$-dimensional. In this case $T|_{U'}$ is essentially the matrix $Tu \in V \otimes W$ and, since for matrices rank and border rank coincide, we have

\begin{corollary}
\label{cor-substitution}
    $\underline{R}(T) \geq n - 1 + m(T)$ where $m(T) = \min\limits_{u \in U\setminus\{0\}} \operatorname{rk} (Tu)$.
\end{corollary}

\subsection{Border rank of the easy Coppersmith-Winograd tensor}

In \cite{coppersmith-winograd-1990-matrix}, Coppersmith and Winograd first describe a simplified version of the main construction. This ``easy version'' uses the tensor
\begin{equation*}
T_{cw} = \sum_{i = 1}^q
(e^{[0]} \otimes e^{[1]}_i \otimes e^{[1]}_i + 
e^{[1]}_i \otimes e^{[0]} \otimes e^{[1]}_i + 
e^{[1]}_i \otimes e^{[1]}_i \otimes e^{[0]}),
\end{equation*}
with $q \ge 2$, which we will call \emph{easy Coppersmith-Winograd tensor}.

The easy Coppersmith-Winograd tensor is a restriction of the full Coppersmith-Winograd tensor obtained using the projection along $e^{[2]}$ onto $\left<e^{[0]}, e^{[1]}_i\right>$, so its border rank is at most $q + 2$.
It is known that this is the exact value of $\underline{R}(T_{cw})$ (see~\cite[Exercise 15.14(3)]{buergisser-1997-book}).

We can write a bilinear map equivalent to $T_{cw}$ in terms of the algebra $A_{CW}$ described in~\S\ref{section-cw}. Let $X$ be the subspace of $A_{CW}$ spanned by $x_i$, $M$ be the subspace spanned by $1$ and $X$, and $R$ be the radical of $A_{CW}$ (the subspace spanned by $x_i$ and $x_1^2$). Denote by $\rho$ the projection of $A_{CW}$ onto $R$ along $1$. Then $T_{cw}$ is equivalent to the bilinear map $\varphi_{cw} \in M^* \otimes M^* \otimes R$ defined as $\varphi_{cw}(a, b) = \rho(ab)$ (the multiplication is in $A_{CW}$).

\begin{lemma}
    Let $q \geq 2$. For any $\psi \in U^* \otimes V^* \otimes W$, we have $m(\varphi_{cw} \otimes \psi) \geq 2m(\psi)$.
\end{lemma}
\begin{proof}
    For a bilinear map $\psi$, the value $m(\psi)$ is the minimum dimension of the space $\psi(u, V)$ among all nonzero $u \in U$.
    
    Consider a nonzero element $a = 1 \otimes u_0 + \sum_{i = 1}^q x_i \otimes u_i \in M \otimes U$.
    If all $u_i = 0$, then \[(\varphi_{cw} \otimes \psi)(a, M \otimes V) = (\varphi_{cw} \otimes \psi)(1 \otimes u_0, M \otimes V) = \varphi_{cw}(1, M) \otimes \psi(u_0, V) = X \otimes \psi(u_0, V)\] has dimension at least $qm(\psi)$.
    Otherwise, without loss of generality assume $u_1 \neq 0$.
    The space $(\varphi_{cw} \otimes \psi)(a, M \otimes V)$ contains subspaces
    \[\begin{aligned}
    S_0 &= (\varphi_{cw} \otimes \psi)(a, 1 \otimes V) = 
    \{\sum_{i = 1}^q x_i \otimes \psi(u_i, v) \mid v\in V\} \\
    S_1 &= (\varphi_{cw} \otimes \psi)(a, x_1 \otimes V) =
    \{ x_1 \otimes \psi(u_0, v) + x_1^2 \otimes \psi(u_1, v) \mid v\in V \}
    \end{aligned}\] which have at least $2m(\psi)$ linearly independent elements, namely, for each of at least $m(\psi)$ linearly independent vectors $z_k \in \psi(u_1, V)$ we have $x_1 \otimes z_k + x_2 \otimes w_2 + \dots + x_q \otimes w_q \in S_0$ and $x_1^2 \otimes z_k + x_1 \otimes w_1 \in S_1$ for some $w_1, w_2, \dots, w_q \in W$.
    
    In both cases we have $\dim (\varphi_{cw} \otimes \psi)(a, M\otimes V) \geq 2m(\psi)$ for all $a \in M\otimes U$.
\end{proof}

\begin{corollary}
    $\underline{R}(T_{cw}^{\otimes n}) \geqslant (q + 1)^n + 2^n - 1$.
\end{corollary}
\begin{proof}
    Use the previous Lemma to show that $m(T_{cw}^{\otimes n}) = m(\varphi_{cw}^{\otimes n}) = 2^n$ and Corollary~\ref{cor-substitution}.
\end{proof}

If $\lim_{n \to \infty} (\underline{R}(T_{cw}^{\otimes n}))^{1/n} = q+1$, then
the exponent of matrix multiplication would be $2$. While the bound above is nontrivial,
it is yet not strong enough to rule this out.

\subparagraph*{Acknowledgements.}

The authors thank Dmitry Chistikov and anonymous referees for helpful comments and Charilaos Zisopoulos for proofreading.

This work was partially supported by the Deutsche Forschungsgemeinschaft under grant BL~511/10-1.

\bibliography{20-blaeser-lysikov}

\end{document}